\documentclass[preprint,11pt]{elsarticle}

\usepackage{algorithm}
\usepackage[noend]{algorithmic}
\usepackage{algorithmicext}
\usepackage{ifthen}
\newconstruct{\PROC}{\textbf{procedure}}{}{\ENDPROC}{\textbf{end procedure}}
\newconstruct{\FUNC}{\textbf{function}}{}{\ENDFUNC}{\textbf{end function}}
\usepackage{fullpage}
\usepackage{lmodern}
\usepackage{comment}
\usepackage{amsthm,amsmath,amssymb}
\usepackage{graphicx}

\usepackage{pstricks,pst-node,pst-tree}
\usepackage{soul}

\usepackage{color}
\usepackage{mathrsfs }
\usepackage[normalem]{ulem}
\usepackage{paralist}
\newcommand{\msg}[1]{\langle#1\rangle}
\renewcommand{\epsilon}{\varepsilon}
\newcommand{\Prob}[1]{\hbox{\rm I\kern-2pt P}\left[#1\right]}

\DeclareMathAlphabet{\mathsc}{OT1}{cmr}{m}{sc}

\newtheorem{theorem}{Theorem}
\newtheorem{lemma}{Lemma}
\newtheorem{corollary}{Corollary}

\renewcommand{\ge}{\geqslant}
\renewcommand{\leq}{\leqslant}
\renewcommand{\le}{\leqslant}
\newcommand{\set}[1]{\{#1\}}

\newcommand{\eps}{\varepsilon}

\def\elected{\mbox{\small ELECTED}}
\def\nonelected{\mbox{\small NON-ELECTED}}

\newboolean{short}
\setboolean{short}{false}

\newcommand{\shortOnly}[1]{\ifthenelse{\boolean{short}}{#1}{}}
\newcommand{\onlyShort}[1]{\ifthenelse{\boolean{short}}{}{#1}}
\newcommand{\longOnly}[1]{\ifthenelse{\boolean{short}}{}{#1}}
\newcommand{\onlyLong}[1]{\ifthenelse{\boolean{short}}{}{#1}}

\newcommand{\squishlist}{
 \begin{list}{$\bullet$}
  { \setlength{\itemsep}{0pt}
     \setlength{\parsep}{3pt}
     \setlength{\topsep}{3pt}
     \setlength{\partopsep}{0pt}
     \setlength{\leftmargin}{1.5em}
     \setlength{\labelwidth}{1em}
     \setlength{\labelsep}{0.5em} } }

\newcommand{\squishlisttwo}{
 \begin{list}{$\bullet$}
  { \setlength{\itemsep}{0pt}
     \setlength{\parsep}{0pt}
    \setlength{\topsep}{0pt}
    \setlength{\partopsep}{0pt}
    \setlength{\leftmargin}{2em}
    \setlength{\labelwidth}{1.5em}
    \setlength{\labelsep}{0.5em} } }

\newcommand{\squishend}{
  \end{list}  }

\begin{document}
\title{Sublinear Bounds for Randomized Leader Election}
\author[technion]{Shay Kutten\fnref{fn1}}
\ead{kutten@ie.technion.ac.il}

\author[ntu,brown]{Gopal Pandurangan\fnref{fn2}}
\ead{gopalpandurangan@gmail.com}
\author[weiz]{David Peleg\fnref{fn3}}
\ead{david.peleg@weizmann.ac.il}
\author[ntu]{Peter Robinson\fnref{fn4}}
\ead{peter.robinson@ntu.edu.sg}
\author[technion]{Amitabh Trehan\fnref{fn1}}
\ead{amitabh.trehaan@gmail.com}

\address[technion]{Information Systems Group, Faculty of Industrial 
  Engineering and Management, Technion - Israel Institute of Technology, 
  Haifa-32000, Israel.}

\address[ntu]{Division of Mathematical
  Sciences, Nanyang Technological University, Singapore 637371}

\address[brown]{Department of Computer Science, Brown University, Box 1910, Providence,
RI 02912, USA.}

\address[weiz]{Department of Computer Science and Applied Mathematics, Weizmann Institute of Science, Rehovot-76100 Israel.}

\fntext[fn1]{Supported by the Israeli Science Foundation and by the Technion TASP center.}
\fntext[fn2]{Research supported in part by the following  grants: Nanyang
Technological University grant M58110000, Singapore Ministry of Education
(MOE) Academic Research Fund (AcRF) Tier 2 grant MOE2010-T2-2-082,
and a grant from the US-Israel Binational Science
Foundation (BSF).}
\fntext[fn3]{Supported in part by the Israel Science Foundation (grant 894/09),
the United States-Israel Binational Science Foundation (grant 2008348),
and the Israel Ministry of Science and Technology (infrastructures grant).}
\fntext[fn4]{Research supported in part by the following  grants: Nanyang Technological University grant M58110000, Singapore Ministry of Education (MOE) Academic Research Fund (AcRF) Tier 2 grant MOE2010-T2-2-082.}

\begin{abstract}
This paper concerns {\em randomized} leader election in synchronous distributed networks. A  distributed leader election algorithm  is presented for complete $n$-node networks that  runs in $O(1)$ rounds and (with high probability) uses only $O(\sqrt{n}\log^{3/2} n)$  messages to elect a unique leader (with high probability). 
When considering the ``explicit'' variant of leader election where eventually every node knows the identity of the leader, our algorithm yields the asymptotically optimal bounds of $O(1)$ rounds and $O(n)$ messages.
This algorithm is then extended to one solving leader election on any connected non-bipartite $n$-node graph $G$ in $O(\tau(G))$ time and $O(\tau(G)\sqrt{n}\log^{3/2} n)$ messages, where $\tau(G)$ is the mixing time of a random walk on $G$. The above result implies highly efficient (sublinear running time and messages)
leader election algorithms for networks with small mixing times,
such as expanders and hypercubes.
In contrast, previous  leader election algorithms had 
at least linear  message complexity
even in complete graphs.
Moreover, super-linear message lower bounds are known for time-efficient {\em deterministic} leader election algorithms.
 Finally, we present an almost matching lower bound for randomized leader election, showing that $\Omega(\sqrt n)$ messages are needed for  any leader election algorithm that
 succeeds with probability at least $1/e + \eps$, for any small constant $\eps > 0$. 
 We view our results as a step towards understanding the randomized complexity of
   leader election in distributed networks.
\end{abstract}
\maketitle

\section{Introduction}

\subsection{Background and motivation}
Leader election is a classical and  fundamental problem in distributed computing. It originated as the problem of regenerating the ``token'' in a local area \emph{token ring} network~\cite{Lann77-DistSystems} and has since then ``starred'' in major roles in problems across the spectrum, %
providing solutions for reliability by replication (or duplicate elimination), for locking, synchronization, load balancing, maintaining group memberships and establishing communication primitives. As an example, the content delivery network giant Akamai uses decentralized and distributed leader election as a subroutine to tolerate machine failure and build fault tolerance in its systems~\cite{Nygren-Akamai}. In many cases, especially with the advent of large scale networks such as peer-to-peer systems~\cite{Ratnasamy01CAN,Rowstron01Pastry,Zhao04Tapestry}, it is desirable to achieve low cost and scalable leader election, even though the guarantees may be probabilistic.

Informally, the problem of distributed leader election requires a group of processors in a distributed network to elect a unique leader among themselves, 
i.e., exactly one processor must output the decision that it is the leader, 
say, by changing a special \emph{status} component of its state to the value 
\emph{leader}~\cite{Lyn96}.
All the rest of the nodes must change their status component to the value \emph{non-leader}.  These nodes need not be aware of the identity of the leader. This {\em implicit} variant of leader election is rather standard (cf.
\cite{Lyn96}), and is sufficient in many applications, e.g., for token generation in a token ring environment. This  paper focuses  on  implicit leader election (but improves the upper bounds also for the explicit case, by presenting a time and message optimal randomized protocol).

In another variant,  all the non-leaders change their status component to the value \emph{non-leader}, and moreover, every node must also know the identity of the unique leader. This formulation may be necessary in problems where nodes coordinate and communicate through a leader, e.g., implementations of Paxos~\cite{Chandra-PaxosTalk-PODC07,Lamport-Paxos98}.
 In this variant, there is an obvious lower bound of $\Omega(n)$ messages (throughout, $n$ denotes the number of nodes in the network) since every node must be informed of the leader's identity. This {\em explicit} leader election
 can be achieved by simply executing an (implicit) leader election algorithm and then broadcasting the leader's identity using an additional $O(n)$ messages and $O(D)$ time (where $D$ is the diameter of the graph).
 
The complexity of the leader election problem and algorithms for it, 
especially deterministic algorithms (guaranteed to always succeed), 
have been well-studied. Various algorithms and lower bounds are known 
in different models with synchronous/asynchronous communication and 
in networks of varying topologies such as a cycle, a complete graph,
 or some arbitrary topology (e.g., see
\cite{KhanKMPT08,Lyn96,peleg-jpdc,santoro-book,GerardTelDistributedAlgosBook} 
and the references therein). The problem was first studied in context of a 
ring network by Le~Lann~\cite{Lann77-DistSystems} and discussed for 
general graphs in the influential paper of Gallager, Humblet, and 
Spira~\cite{GallagerHS1983}. However, leader election in the class of 
complete networks has come to occupy a special position of its own 
and has been extensively studied~\cite{AG1991:SICOMP,humblet-clique,KorachKuttenMoran-ModularLE-TOPLAS,KorachOptimalTrees87,KorachOptimal89,singh}; see also \cite{sense-of-dir,LouiMW88,singh97} for leader election in complete networks where nodes have a sense of direction.

The study of leader election algorithms is usually concerned with both message and time complexity. For complete graphs, Korach et al.~\cite{KorachPODC1984} and Humblet \cite{humblet-clique} presented $O(n \log n)$ message algorithms.
  Korach, Kutten, and Moran~\cite{KorachKuttenMoran-ModularLE-TOPLAS} developed a general method  decoupling the issue of the graph family from the design of the leader election algorithm, allowing the development of message efficient leader election algorithms for any class of graphs, given an efficient traversal algorithm for that class. When this method was applied to complete graphs, it yielded an improved (but still
  $\Omega(n \log n)$) message complexity. Afek and Gafni~\cite{AG1991:SICOMP} presented asynchronous and synchronous algorithms, as well as a tradeoff between the message and the time complexity of synchronous {\em deterministic} algorithms for complete graphs: the results varied from a $O(1)$-time, $O(n^2)$-messages  algorithm to a $O(\log n)$-time, $O(n\log n)$-messages algorithm.
Singh \cite{singh} showed another trade-off that saved on time, still for 
algorithms with a super-linear number of messages. (Sublinear time algorithms were shown in \cite{singh} even for $O(n \log n)$ messages algorithms, and even lower times for algorithms with higher messages complexities).
Afek and Gafni, as well as~\cite{KorachPODC1984,KorachOptimal89} showed a lower bound of $\Omega(n \log n)$ messages  for {\em deterministic} algorithms in the general case.
 One specific case where the message complexity could be reduced (but only as far as linear message complexity) was at the expense of also having a
  linear time complexity, see \cite{AG1991:SICOMP}.
   Multiple studies showed a different case where it was possible to reduce the number of messages to $O(n)$, by using a {\em sense of direction} - essentially, assuming some kind of a virtual ring, superimposed on the complete graph, such that the order of nodes on a ring is known to the nodes \cite{sense-of-dir}.
The above results demonstrate that the number of messages 
needed for deterministic leader election is at least linear or even 
super-linear (depending on the time complexity). 
In particular, existing $O(1)$ time deterministic algorithms require
$\Omega(n^2)$ messages (in a complete network).
At its core, leader election is a symmetry breaking problem. For anonymous networks under some reasonable assumptions, deterministic leader election was shown to be impossible~\cite{AngluinSTOC80} (using symmetry arguments). Randomization comes to the rescue in this case; random rank assignment is often used to assign unique identifiers, as done herein. Randomization also allows us to beat the lower bounds for deterministic algorithms, albeit at the risk of a small chance of error. 

A randomized leader election algorithm (for the explicit version) that could 
err with probability $O(1 / \log^{\Omega(1)}n )$ 
was presented in \cite{KrishnaRamanathan:randomized}
with time $O(\log n)$ and linear message complexity\footnote{In contrast, 
the probability of error in the current paper is $O(1 / n^{\Omega(1)})$.}. 
That paper also surveys some related papers about randomized algorithms in 
other models that use more messages for performing leader election
~\cite{Gupta-ProbLE-DISC} or related tasks (e.g., probabilistic quorum systems,
 Malkhi et al~\cite{Malkhi-ProbQSystems}).
    In the context of self-stabilization, a randomized
 algorithm with $O(n \log n)$ messages and $O(\log n)$ time until stabilization was presented in \cite{dmitry}.

\subsection{Our Main Results}
The main focus of this paper is on studying how randomization can help in 
improving the complexity of leader election,
especially message complexity in synchronous networks. 
We first present an (implicit) randomized leader election algorithm for a 
complete network that runs in $O(1)$ time and uses only $O(\sqrt{n}\log^{3/2}n)$
messages to elect a unique leader with high probability\footnote{Throughout, 
``with high probability (w.h.p)" means with probability at least 
$1 - 1/n^{\Omega(1)}$.}. 
This is a significant improvement over the linear number of messages that is 
needed for any deterministic algorithm. It is an even larger improvement over 
the super-linear number of messages needed for deterministic algorithms 
that have low time complexity (and especially compared to the $O(n^2)$ 
messages for deterministic $2$-round algorithms).
For the explicit variant of the problem, our algorithm implies an algorithm 
that uses (w.h.p.) $O(n)$ messages and $O(1)$ time, still a significant 
improvement over the $\Omega(n^2)$ messages used by deterministic algorithms.

We then extend this algorithm to solve leader election on any connected 
(non-bipartite\footnote{Our algorithm can be easily modified to work 
for bipartite graphs as well --- cf. Section 3.}) $n$-node graph $G$ 
in $O(\tau(G))$ time and $O(\tau(G)\sqrt{n}\log^{3/2} n)$ messages, 
where $\tau(G)$ is the mixing time of a random walk on $G$. 
The above result implies highly efficient (sublinear running time and messages)
leader election algorithms for networks with small mixing time.
In particular, for  important graph classes such as expanders (used, e.g., 
in modeling peer-to-peer networks \cite{soda12}), which have a logarithmic 
mixing time, it implies an algorithm of $O(\log n)$ time and 
$O(\sqrt{n}\log^{5/2} n)$ messages, and for hypercubes, which have 
a mixing time of $O(\log n \log \log n)$, it implies an algorithm of
$O(\log n \log \log n)$ time and $O(\sqrt{n}\log^{5/2} n \log \log n)$ messages.

For our algorithms, we assume that the communication is synchronous and follows the standard $\mathcal{CONGEST}$ model~\cite{Pel00}, where a node can send in each round at most one message of size $O(\log n)$ bits on a single edge. For our algorithm on general graphs,  we also assume that the nodes have an estimate of the network's size (i.e., $n$)
and the mixing time. We do not however assume that the nodes have unique IDs, hence the algorithms in this paper apply also for anonymous networks.  We assume that all nodes wake up simultaneously at the beginning of the execution. (Additional details on our distributed computation model are given later on.)

Finally we show that, in general, it is not possible to improve over 
our algorithm substantially, by presenting a lower bound for randomized 
leader election. We show that $\Omega(\sqrt n)$ messages are needed for any 
leader election algorithm in a complete network which succeeds with 
probability at least $1/e + \eps$ for any constant $\eps > 0$.
This lower bound holds even in the $\mathcal{LOCAL}$ model~\cite{Pel00}, 
where there is no restriction on the number of bits that can be sent 
on each edge in each round. 
To the best of our knowledge, this is the first non-trivial lower bound
for randomized leader election in complete networks.

\subsection{Technical Contributions}

The main algorithmic tool used by our randomized algorithm involves 
reducing the message complexity via random sampling.
For general graphs, this sampling is implemented by performing random 
walks.
Informally speaking, a small number of nodes (about $O(\log n)$), which 
are the candidates for leadership, initiate random walks.  We show that if 
sufficiently many random walks are initiated (about $\sqrt{n}\log 
n$), then there is a
good probability that  random walks originating from different candidates meet 
(or collide) at some node which acts as a referee. 
The referee notifies a winner among the colliding random walks.  The 
algorithms use a birthday paradox type argument to show that a unique 
candidate node wins all competitions (i.e., is elected) with high 
probability.
An interesting feature of that birthday paradox argument (for general graphs) 
is that it is applied to a setting with non-uniform selection probabilities. 
See Section 2 for a simple version of the  algorithm that works 
on a complete graph.
The algorithm of Section 3 is a generalization of the algorithm of 
Section 2 that works for any connected graph; however the algorithm and 
analysis are more involved.

The main intuition behind our lower bound proof for randomized leader election
is that, in some precise technical sense, any algorithm that sends fewer 
messages than required by our lower bound has a good chance of generating runs 
where there are multiple potential leader candidates in the network that 
do not influence each other.
In other words, the probability of such ``disjoint'' parts of the network 
to elect a leader is the same, which implies that there is a good 
probability that more than one leader is
elected. Although this is conceptually easy
to state, it is technically challenging to show formally since our result 
applies to all randomized algorithms without further restrictions.

%

%

 %

%
%
%
%
%
%
%

%

%
%
%
\subsection{Distributed Computing Model} \label{sec:model}
The model we consider is similar to the models of~\cite{AG1991:SICOMP,humblet-clique,KorachKuttenMoran-ModularLE-TOPLAS,KorachOptimalTrees87,KorachOptimal89},
with the main addition of giving processors access to a private unbiased coin. 
Also, we do not assume unique identities.
We consider a system of $n$ nodes, represented as an undirected (not
necessarily complete) graph $G=(V,E)$.
Each node %
runs an instance of a distributed algorithm.
The computation advances in synchronous rounds where, in every round, nodes
can send messages, receive messages that were sent in the same round by
neighbors in $G$,
and perform some local computation.
Every node has access to the outcome of unbiased private coin flips.
Messages are the only means of communication; in particular, nodes
cannot access the coin flips of other nodes, and do not share any memory.
Throughout this paper, we assume that all nodes are awake initially and 
simultaneously start executing the algorithm. 
\subsection{Leader Election}
We now formally define the leader election problem.
Every node $u$ has a special variable $\texttt{status}_u$ that it can set
to a value in $\{\bot, \nonelected, \elected \}$; initially we assume
$\texttt{status}_u = \bot$.
An \emph{algorithm $A$ solves leader election in $T$ rounds}
if, from round $T$ on, exactly one node has its status set to $\elected$ 
while all other nodes are in state $\nonelected$. This is the requirement 
for standard (implicit) leader election.  
%

%

\section{Randomized Leader Election in Complete Networks}
\label{sec: complete}

To provide the intuition for our general result, let us start by illustrating a
simpler version of our leader election algorithm, adapted to complete networks.
More specifically, this section presents an algorithm that, with high
probability, solves leader election in complete networks in $O(1)$ rounds
and sends no more than $O(\sqrt{n}\log^{3/2}n)$ messages.
Let us first briefly describe the main ideas of
Algorithm~\ref{alg:leaderComplete} (see pseudo-code below).
Initially, the algorithm attempts to reduce the number of leader candidates
as far as possible, while still guaranteeing that there is at least one
candidate (with high probability).
Non-candidate nodes enter the $\nonelected$ state immediately,
and thereafter only reply to messages initiated by other nodes.
Every node $u$ becomes a candidate with probability $2\log n / n$
and selects a random rank $r_u$ chosen from some large domain.
Each candidate node then randomly selects $2\lceil\sqrt{n\log n}\rceil$
other nodes as {\em referees} and informs all referees of its rank.
The referees compute the maximum (say $r_w$) of all received ranks, and
send a ``winner'' notification to the node $w$.
If a candidate wins all competitions, i.e., receives ``winner'' notifications
from all of its 
referees, it enters the $\elected$ state and becomes the leader.

\begin{algorithm}[h!]
\begin{algorithmic}[1]
\item[]
 \item[\bf Round 1:]
  \STATE Every node $u$ decides to become a candidate with probability
  $2\log n / n$
  and generates a random rank $r_u$ from $\set{1,\dots,n^4}$.
  \\
  If a node $u$ does not become a candidate, then it
  immediately enters the $\nonelected$ state;
  otherwise it executes the next step.
  \STATE \textbf{Choosing Referees:}
  Node $u$ samples $2\lceil\sqrt{n\log n}\rceil$ neighbors (the
  \emph{referees}) and sends a message $\msg{u,r_u}$ to each referee.

  \item[] 

  \item[\bf Round 2:]
  \STATE  \textbf{``Winner'' Notification:}
  A referee $v$ considers all received messages and sends
  a ``winner'' notification to the node $w$ of maximum rank, 
  namely, that satisfies 
  $r_w\ge r_u$ for every message $\msg{u,r_u}$.
    \STATE \textbf{Decision:} 
     If a node 
     receives ``winner'' notifications from all its referees, 
     then it enters the $\elected$ state; 
\\
     otherwise it sets its state to $\nonelected$.
\end{algorithmic}
\caption{Randomized Leader Election in Complete Networks}
\label{alg:leaderComplete}
\end{algorithm}

\begin{theorem} \label{thm:leaderComplete}
Consider a complete network of $n$ nodes
and assume the $\mathcal{CONGEST}$ model of communication.
With high probability, Algorithm~\ref{alg:leaderComplete} solves leader
election in $O(1)$ rounds, while using $O(\sqrt{n}\log^{3/2} n)$ messages.
\end{theorem}
\begin{proof}%
  Since all nodes enter either the $\elected$ or $\nonelected$
  state after two rounds at the latest, the runtime bound of $O(1)$
  holds trivially.

  We now argue the message complexity bound.
  On expectation, there are $2\log n$ candidate nodes.
  By using a standard Chernoff bound (cf.\ Theorem~4.4 in \cite{upfal}),
  there are at most $7\log n$ candidate nodes with probability at least
  $1-n^{-2}$.
  In step 3 of the algorithm, each referee only sends
  messages to the candidate nodes which contacted it.
  Since there are $O(\log n)$ candidates and each approaches  
  $2\lceil\sqrt{n\log n}\rceil$ referees, the total number of messages sent
  is bounded by $O(\sqrt{n}\log^{3/2} n)$ with high probability.

  Finally, we show that Algorithm 1    %
  solves leader election with high probability.
  The probability that no node elects itself as leader is
  $$\left(1 - \frac{2\lceil\log n\rceil}{n}\right)^n \approx 
    \exp(-2\log n) = n^{-2}.$$
  Hence the probability that at least
  one node is elected as leader is at least $1-n^{-2}$.
  Let $\ell$ be the node that generates the highest random rank $r_\ell$
  among all candidate nodes; with high probability, $\ell$ is unique.
  Clearly, node $\ell$ enters the $\elected$ state, since it
  receives ``winner'' notifications from all its referees.

  Now consider some other candidate node $v$. This candidate
  chooses its referees randomly among all nodes.
  Therefore, the probability that an individual referee selected by $v$
  is among the referees chosen by $\ell$, is
  $2\lceil\sqrt{n \log n}\rceil/n$.
  It follows that the probability that $\ell$ and $v$ do not choose any
  common referee node is asymptotically at most
  $$  \left( 1 - 2\sqrt{\frac{\log n}{n}}\right)^{2\sqrt{n\log n}} \le
  \exp\left(-4\log n\right) = n^{-4},  $$
  which means that with high probability, some node $x$ serves as
  common referee to $\ell$ and $v$.
  By assumption, we have $r_v < r_\ell$, which means that node $v$ does
  not receive $2\lceil\sqrt{n\log n}\rceil$ ``winner'' notifications,
  and thus it subsequently enters the $\nonelected$ state.
  By taking a union bound over all candidate nodes other than $\ell$,
  it follows that with probability at least $1-1/n$,
  no other node except $\ell$ wins all of its competitions,
  and therefore, node $\ell$ is the only node to become a leader.
\end{proof}

\section{Randomized Leader Election in General Graphs}
\label{sec: general}
In this section, we present our main algorithm, which elects a unique 
leader in $O(\tau)$ rounds (w.h.p.), while using 
$O(\tau(G,n)\sqrt{n}\log^{3/2} n)$ messages (w.h.p.), where $\tau(G,n)$ is 
the mixing time of a random walk on $G$ (formally defined later on, 
in Eq. \eqref{eq:mixing}).
Initially, any node $u$ only knows the mixing time (or a constant factor 
estimate of) $\tau(G,n)$; in particular $u$ does
not have any a priori knowledge about the actual topology of $G$.

The algorithm presented here requires nodes to perform random walks on the
network by token forwarding in order to choose sufficiently many referee
nodes at random. Thus essentially random walks perform the role of sampling
as done in Algorithm \ref{alg:leaderComplete} and is conceptually similar.
Whereas in the complete graph randomly chosen nodes act as referees,
here any intermediate node (in the random walk) that sees tokens
from two competing candidates can act as a referee and notify the winner.
One slight complication we have to deal with in the general setting is that 
in the $\mathcal{CONGEST}$ model it is impossible 
to perform too many walks in parallel along an edge. 
We solve this issue by sending only the {\em count} of tokens that need 
to be sent by a particular candidate, and not the tokens themselves.

While using random walks can be viewed as a generalization of the sampling
performed in Algorithm~\ref{alg:leaderComplete}, showing that two
candidate nodes intersect in at least one referee 
leads to an interesting balls-into-bins scenario where balls 
(i.e., random walks) have a \emph{non-uniform} probability to be placed 
in some bin (i.e., reach a referee node).
This non-uniformity of the random walk distribution stems from the fact
that $G$ might not be a regular graph.
We show that the non-uniform case does not worsen the
probability of two candidates reaching a common referee, and hence an analysis
similar to the one given for complete graphs goes through.

We now introduce some basic notation for random walks.
Suppose that $V=\{u_1,\dots,u_n\}$ and let $d_i$ denote the degree of node
$i$.
The $n\times n$ \emph{transition matrix} $\mathbf{A}$ of $G$ has entries
$a_{i,j}=1/d_i$ 
if there is an edge $(i,j) \in E$, otherwise
$a_{i,j}=0$.
The entry $a_{i,j}$ gives the probability that a random walk moves from node
$u_i$ to node $u_j$.
The position of a random walk after $k$ steps is represented by a
probability distribution $\pi_k$ determined by $\mathbf{A}$.
If some node $u_i$ starts a random walk, the initial distribution $\pi_0$ of
the walk is an $n$-dimensional vector having all zeros except at index $i$
where it is $1$.
Once the node $u$ has chosen a random neighbor to forward the token, the
distribution of the walk after $1$ step is given by $\pi_1 = \mathbf{A}\pi_0$
and in general we have $\pi_k = \mathbf{A}^k \pi_0$.
If $G$ is non-bipartite and connected, then the distribution of the walk will
eventually converge to the \emph{stationary distribution}
$\pi_*=(b_1,\dots,b_n)$, which has entries $b_i = d_i / (2|E|)$
and satisfies $\pi_* = \mathbf{A}\pi_*$.

We define the \emph{mixing time} $\tau(G,n)$ of a graph G with $n$ nodes
as the minimum $k$ such that, for all starting distributions $\pi_0$,
\begin{equation} \label{eq:mixing}
  |\!| \mathbf{A}\pi_k - \pi_* |\!|_{\infty} \le \frac{1}{2n},
\end{equation}
where $|\!|\cdot|\!|_{\infty}$ denotes the usual maximum norm on a vector.
Clearly, if $G$ is a complete network, then
$\tau(G,n)=1$.
For expander graphs  it is well known that $\tau(G,n) \in O(\log n)$.
Note that mixing time is well-defined only for non-bipartite graphs;
however, by using a lazy random walk strategy (i.e., with probability $1/2$
stay in the current node; otherwise proceed as usual) our algorithm
will work for bipartite graphs as well.

\begin{algorithm}[h!]
\begin{algorithmic}[1]
  \item[]
  \item[\textbf{Variables and Initialization:}]
  \STATE \textbf{VAR} $\texttt{origin} \leftarrow 0$; $\texttt{winner-so-far} \leftarrow \bot$
  \STATE Node $u$ decides to become a candidate with probability
  $2\log n / n$ 
  and generates a random rank $r_u$ from $\set{1,\dots,n^4}$.
  \item[]
  \item[\textbf{Initiating Random Walks:}]
  \STATE Node $u$ creates $2\lceil\sqrt{n\log n}\rceil$ tokens of type 
  $\msg{r_u,k}$.
  \STATE Node $u$ starts $2\lceil\sqrt{n\log n}\rceil$ random walks 
  (called \emph{competitions}), each of
  which is represented by the random walk token
  $\msg{r_u,k}$ (of
  $O(\log n)$ bits) where $r_u$ represents $u$'s random rank.
  The counter $k$ is the number (initially 1) of walks
  that are represented by this token (explained in
  Line~\ref{line:counter}).

  \item[]
  \item[\textbf{Disqualifying low-rank candidates:}] (note that any 
   intermediate node along the random walk can act as a referee and 
   disqualify the token of a low-rank candidate)
  \STATE \label{line:disqualify} A node $v$ discards every received token
  $\msg{r_u,k}$
  if $v$ has received (possibly in the same round) a token $r_w$
    with $r_w>r_u$.
  \IF{a received token $\msg{r_w,k'}$ is not discarded and
  $\texttt{winner-so-far} \neq r_w$}
  \STATE Node $v$ remembers the port of an arbitrarily chosen neighbor that sent
  one of the (possibly merged) tokens containing $r_w$ in variable $\texttt{origin}$
  and sets its variable $\texttt{winner-so-far}$ to $r_w$. \label{line:origin}
  \ENDIF
  \item[]
  \item[\textbf{Token Forwarding:}]
  \STATE Let $\mu=\msg{r_u,k}$ be a token received by $v$ and
  suppose that $\mu$ is not discarded in Line~\ref{line:disqualify}.
  For simplicity, we consider all distinct tokens that arrive in the
  current round containing the same value $r_u$ at $v$ to be merged into a
  single token $\msg{r_u,k}$ before processing where $k$ holds the accumulated
  count.
  Node $v$ randomly samples $k$ times from its neighbors.
  \label{line:counter}
  If a neighbor $x$ was chosen $k_x\le k$ times, $v$ sends a token
  $\msg{r_u,k_x}$ to $x$.

  \item[]
  \item[\textbf{Notifying a winner in round $\tau(G,n)$:}]
  \IF{$\texttt{winner-so-far} \neq \bot$}
  \STATE Suppose that node $v$ has not discarded some token generated by
  a node $w$.
  According to Line~\ref{line:disqualify}, $w$ has generated the largest
  rank among all tokens seen by $v$. %
  \STATE Node $v$ generates a ``winner'' notification 
   $\msg{\textsc{WIN},r_w,cnt}$
  for $r_w$ and sends it to the neighbor stored in $\texttt{origin}$ (cf.\
  Line~\ref{line:origin}). The field $cnt$ is set to $1$ by $v$ and contains the
  number of ``winner'' notifications represented by this token.
  \ENDIF
  \STATE If a node $u$ receives (possibly) multiple ``winner'' notifications 
    for $r_w$, it simply
  forwards  a token $\msg{\textsc{WIN},r_w,cnt'}$ to the neighbor stored in
  $\texttt{origin}$ where $cnt'$ is the accumulated count of all received tokens.

  \item[]
  \item[\textbf{Decision:}]
  \STATE If a node wins all competitions, i.e., receives $2\lceil\sqrt{n\log 
    n}\rceil$
  ``winner'' notifications it enters the $\elected$ state;
  otherwise it sets its state to $\nonelected$.

\end{algorithmic}
\caption{Randomized Leader Election in General Networks}
\label{alg:leader}
\end{algorithm}

\begin{theorem} \label{thm:leader}
Consider a non-bipartite network $G$ of $n$ nodes with mixing time
$\tau(G,n)$, and assume the $\mathcal{CONGEST}$ model of communication.
With high probability, Algorithm~\ref{alg:leader} solves leader election
within $O(\tau(G,n))$ rounds, while using $O(\tau(G,n)\sqrt{n}\log^{3/2}
n)$ messages.
\end{theorem}

\begin{proof}%
  We first argue the message complexity bound.
  As argued in the proof of Thm. \ref{thm:leaderComplete},
  there are at most $7\log n$ candidate nodes with probability at least
  $1-n^{-2}$.
  Every candidate node $u$ 
  creates $\Theta(\sqrt{n\log n})$ tokens
  and initiates a random walk of length $\tau(G,n)$, 
  for each of the $\Theta(\sqrt{n\log n})$
  tokens.
  By the description of the algorithm, 
  there are $O(\sqrt{n}\log^{3/2} n)$ random walks of length $O(\tau(G,n))$.
  In addition, at most one notification message is sent at the last step 
  of each random walk, and it travels a distance of at most $O(\tau(G,n))$.
  Hence the total number of messages sent throughout the execution is
  bounded by $O(\tau(G,n)\sqrt{n}\log^{3/2} n)$ with high probability.

  The running time bound depends on the time that it takes to complete the
  $2\lceil\sqrt{n\log n}\rceil$ random walks in parallel and the notification 
  of the winner.
  By Line~\ref{line:disqualify}, it follows that a node only forwards
  at most one token to any neighbor in a round, thus there is no delay due
  to congestion.
  Moreover, for notifying the winner, nodes forward the ``winner'' notification
  for winner $w$ to the neighbor stored in $\texttt{origin}$.
  According to Line~\ref{line:origin}, a node sets $\texttt{origin}$ to a
  neighbor from which it has received the first token originated from $w$.
  Thus there can be no loops when forwarding the ``winner'' notifications,
  which reach the winner $w$ in at most $\tau(G,n)$ rounds.

  We now argue that Algorithm~\ref{alg:leader} solves leader election
  with high probability.
  Similarly to Algorithm~\ref{alg:leaderComplete}, it follows that there
  will be at least one leader with high probability.

  Now consider some other candidate node $v$.
  Recall that we have that $r_v < r_\ell$ by assumption.
  By the description of the algorithm, node $v$ chooses its referees by
  performing $\rho = 2\lceil\sqrt{n\log n}\rceil$ random walks of length 
  $\tau(G,n)$.
  We cannot argue the same way as in the proof of
  Algorithm~\ref{alg:leaderComplete}, since in general, the stationary
  distribution of $G$ might not be the uniform distribution vector
  $(1/n,\dots,1/n)$.
  Let $p_i$ be the $i$-th entry of the stationary distribution.
  Let $X_i$ be the indicator random variable that is $1$ if there is a
  collision (of random walks) at referee node $i$.
  We have
  $$
  \Prob{ X_i = 1 } = (1 - (1-p_i)^\rho )^2.
  $$
  We want to show that the probability of error (i.e., having no collisions)
  is small;
  in other words, we want to upper bound $\Prob{ \bigcap_{i=1}^n (X_i = 0)}$.
  The following Lemma shows that
  $\Prob{\bigcap_{i=1}^n (X_i = 0) }$ is maximized for the uniform
  distribution.
\begin{lemma} \label{lem:poisson}
  Consider $\rho$ balls that are placed into $n$ bins according to some
  probability distribution $\pi$ and let $p_i$ be the $i$-th entry of
  $\pi$.
  Let $X_i$ be the indicator random variable that is $1$ if there is a
  collision (of random walks) at referee node $i$.
  Then $\Prob{ \bigcap_{i=1}^n (X_i = 0) }$ is maximized for the uniform
  distribution.
\end{lemma}
\begin{proof}
  By definition, we have
  $\Prob{ X_i = 1 } = (1 - (1-p_i)^\rho )^2.$
  Note that the events $X_i=1$ and $X_j=1$ are not necessarily
  independent.
  A common technique to treat dependencies in balls-into-bins scenarios is
  the Poisson approximation where we consider the number of balls in each bin
  to be independent Poisson random variables with mean $\rho/n$.
  This means we can apply Corollary 5.11 of \cite{upfal}, which states
  that if some event $E$ occurs with probability $p$ in the Poisson case,
  it occurs with probability at most $2p$ in the exact case, i.e., we only
  lose a constant factor by using the Poisson approximation.
  A precondition for applying Corollary 5.11, is that the probability for
  event $E$ monotonically decreases (or increases) in the number of balls,
  which is clearly the case when counting the number of collisions of
  balls.
  Considering the Poisson case, we get
  \begin{align*}
    \Prob{ \bigcap_{i=1}^n (X_i=0) }
        &= \prod_{i=1}^n \Prob{ X_i = 0 }
        = \prod_{i=1}^n\left( 1 - (1 - (1 - p_i)^\rho)^2\right) \\
        &\leq \prod_{i=1}^n \left(1 - (1 - e^{-p_i \rho})^2 \right)
        \leq \prod_{i=1}^n \left(1 - (p_i \rho)^2 \right) 
        \leq \prod_{i=1}^n e^{-p_i^2 \rho^2}
        = \exp\left(-\rho^2 \sum_{i=1}^n p_i^2\right).
  \end{align*}
  To maximize $\Prob{ \bigcap_{i=1}^n (X_i=0) }$, it is thus sufficient to
  minimize $\sum_{i=1}^n p_i^2$ under the constraint $\sum_{i=1}^n p_i =
  1$.
  Using Lagrangian optimization it follows that this is minimized for the
  uniform distribution, which completes the proof of Lemma~\ref{lem:poisson}.
\end{proof}
  By \eqref{eq:mixing}, the probability of such a walk hitting any of the
  referees chosen by $\ell$, is at least $2\sqrt{n \log n} / (2n)$.
  It follows that the probability that $\ell$ and $v$ do not choose a
  common referee node is asymptotically at most
  \[
  \left( 1 - \sqrt{\frac{\log n}{n}}\right)^{2\sqrt{n\log n}} \le
  \exp\left(-2\log n\right).
  \]
  Therefore, the event that node $v$ does not receive sufficiently many
  ``winner'' notifications, happens with probability $\ge 1-n^{-2}$, which
  requires $v$ to enter the $\nonelected$ state.
  By taking a union bound over all other candidate nodes, it follows that
  with high probability no other node except $\ell$ will win all of its
  competitions, and therefore, node $\ell$ is the only node to become a leader
  with probability at least $1-1/n$.
\end{proof}

\section{Lower Bound}
\label{sec: lower}

In this section, we prove a lower bound on the number of messages required
by any algorithm that solves leader election with probability at least
$1/e + \eps$, for any constant $\eps >0$.

Our model
assumes that all processors execute the same algorithm and have access 
to an unbiased private coin.
So far we have assumed that nodes are \emph{not} equipped with unique ids.
Nevertheless, our lower bound still holds even if the nodes start with
unique ids.

Our lower bound applies to all algorithms that send only $o(\sqrt{n})$
messages with probability at least $1-1/n$.
In other words, the result still holds for algorithms that have small but
nonzero probability for producing runs where the number of messages sent
is much larger (e.g., $\Omega({n})$).
We show the result for the $\mathcal{LOCAL}$ model, which implies the
same for the $\mathcal{CONGEST}$ model.

\begin{theorem} \label{thm:lb}
Consider any algorithm $A$ that sends at most $f(n)$ messages (of
arbitrary size) with high probability on a complete network of $n$ nodes.
If $A$ solves leader election with probability at least $1/e + \eps$, for any
constant $\eps > 0$, then $f(n) \in \Omega(\sqrt{n})$.
This holds even if nodes are equipped with unique identifiers (chosen by
an adversary).
\end{theorem}

Note that Theorem~\ref{thm:lb} is essentially tight with respect to 
the number of messages \emph{and} the probability of successfully electing 
a leader. To see this, 
first observe that our Algorithm~\ref{alg:leaderComplete} can be modified
such that each node becomes a candidate with probability $c/n$, for some 
constant $c>0$, and where each candidate only contacts $\Theta(\sqrt{n})$ 
referee nodes. This yields a message complexity of $O(\sqrt{n})$ and success 
with (large) constant probability.
Furthermore, consider the naive randomized algorithm where each node initially 
chooses to become leader with probability $1/n$ and then terminates.
This algorithm succeeds with probability 
${n \choose 1}(1/n)(1-1/n)^{n-1} \approx 1/e$ without sending any messages 
at all, which demonstrates that there has to be a sudden ``jump'' 
in the required message complexity when breaking the $1/e$ barrier 
in success probability.

The rest of this section is dedicated to proving Theorem~\ref{thm:lb}.
We first show the result for the case where nodes are anonymous, i.e.,
are \emph{not} equipped with unique identifiers,
and later on extend the impossibility result to
the non-anonymous case by an easy reduction.

Assume that there exists some algorithm $A$ that solves leader election with 
probability at least $1/e  + \eps$ but sends only $f(n) \in o(\sqrt{n})$ 
messages.
The remainder of the proof involves showing that this yields a contradiction.
Consider a complete network where for every node, the adversary chooses 
the connections of its ports as a random permutation on $\{1,\dots,n-1\}$.

For a given run $\alpha$ of an algorithm, define the \emph{communication graph}
$\mathcal{C}^r(\alpha)$ to be a directed graph on the
given set of $n$ nodes where there is an edge from $u$ to $v$ if and only
if $u$ sends a message to $v$ in some round $r'\le r$ of the run $\alpha$.
For any node $u$, denote the \emph{state of $u$ in round $r$ of the run 
$\alpha$} by $\sigma_r(u,\alpha)$.
Let $\Sigma$ be the set of all node states possible in algorithm $A$.
(When $\alpha$ is known, we may simply write $\mathcal{C}^r$ and 
$\sigma_r(u)$.)
With each node $u \in \mathcal{C}^r$, associate its state $\sigma_r(u)$ in 
$\mathcal{C}^r$, the communication graph of round $r$.
We say that node \emph{$u$ influences node $w$ by round $r$} if there is a
directed path from $u$ to $w$ in $\mathcal{C}^r$.
(Our notion of influence is more general than the {causality} based 
``happens-before'' relation of \cite{Lam78}, since a directed path from 
$u$ to $w$ is necessary but not sufficient for $w$ to be causally 
influenced by $u$.)
A node $u$ is an \emph{initiator} if it is not influenced before sending 
its first message.
That is, if $u$ sends its first message in round $r$, then $u$ has an outgoing edge in $\mathcal{C}^{r}$ and is an isolated vertex in $\mathcal{C}^{1},\dots,\mathcal{C}^{r-1}$.
For every initiator $u$, we define the \emph{influence cloud} ${\cal 
  IC}^r_u$
as the pair ${\cal IC}^r_u = (C^r_u, S^r_u)$, where
$C^r_u = \langle u, w_1,\ldots,w_k\rangle$ is the ordered set of all nodes 
that are influenced by $u$, namely, that are reachable along a directed path 
in $\mathcal{C}^r$ from $u$. ordered by the time by which they joined\footnote{We say that a node $v$ \emph{joins the cloud of $u$ in $r$} if $v \notin {C}_u^{r-1}$ and $v \in {C}_u^r$.} the cloud (breaking ties arbitrarily),
and $S^r_u = \langle
\sigma_r(u,\alpha), \sigma_r(w_1,\alpha),\ldots, \sigma_r(w_k,\alpha) \rangle$
is their configuration after round $r$, namely, their current tuple of states.
(In what follows, we sometimes abuse notation by referring to the ordered 
node set $C^r_u$ as the influence cloud of $u$.)
Note that a \emph{passive} (non-initiator) node $v$ does not send any messages
before receiving the first message from some other node.

Since we are only interested in algorithms that send a finite number of
messages, in every execution $\alpha$ there is some round
$\rho = \rho(\alpha)$ by which no more messages are sent.

In general, it is possible that in a given execution, two influence clouds
$C^r_{u_1}$ and $C^r_{u_2}$ intersect each other over some common node $v$,
if $v$ happens to be influenced by both $u_1$ and $u_2$.
The following lemma shows that the low message complexity of algorithm $A$
yields a good probability for all influence clouds to be disjoint from
each other.

Hereafter, we fix a run $\alpha$ of algorithm $A$. 
Let ${N}$ be the event that there is no intersection between 
(the node sets of) the influence clouds existing at the end of run $\alpha$, i.e., 
$C^\rho_u \cap C^\rho_{u'} = \emptyset$ for every two initiators $u$ and $u'$.
Let $M$ be the event that algorithm $A$ sends no more than $f(n)$ messages
in the run $\alpha$.

\begin{lemma} \label{lem:noint}
Assume that $\Prob{M} \ge 1-1/n$.
If $f(n) \in o(\sqrt{n})$, then 
$\Prob{N \wedge M} \ge 1 - \frac{1}{n}  - \frac{f^2(n)}{n-f(n)} \in 1-o(1)$.
\end{lemma}
\begin{proof}
Consider a round $r$, some cloud $C^r$ and any node $v \in C^r$.
Assuming event $M$, there are at most $f(n)$ nodes that have sent or received a message and may thus be be a part of some other cloud except $C^r$.
Recall that the port numbering of every node was chosen uniformly at random and,
since we conditioned on the occurrence of event $M$, any node knows the 
destinations of at most $f(n)$ of its ports in any round.
Therefore, to send a message to a node in another cloud, $v$ must hit upon one 
of the (at most $f(n)$) ports leading to other clouds, from among its (at least 
$n-f(n)$) yet unexposed ports. 
Let $H_v^r$ be the event that a message 
sent by node $v$ in round $r$ reaches a node $u$ that is already part of some other (non-singleton) cloud.
(Recall that if $u$ is in a singleton cloud due to not having received or sent 
any messages yet, it simply becomes a member of $v$'s cloud.)
We have $\Prob{H_v^r} \le \frac{f(n)}{n-f(n)}$.
During the entire run, $\ell \le f(n)$ messages are sent in total by some nodes $v_1,\dots,v_\ell$ (in possibly distinct clouds) in rounds $r_1,\dots,r_\ell$, yielding events $H_{v_1}^{r_1},\dots,H_{v_\ell}^{r_\ell}$.
Taking a union bound shows that $$\Prob{\bigvee_{i=1}^\ell H_{v_i}^{r_i} \mid M} \le \frac{f^2(n)}{n-f(n)},$$ which is $o(1)$, for 
$f(n) \in o(\sqrt{n})$. 
Observe that $\Prob{N \mid M} = 1 - \Prob{\bigvee_{i=1}^\ell H_{v_i}^{r_i} \mid M}$. 
Since $\Prob{N\wedge M} = \Prob{N\mid M} \cdot \Prob{M}$,
it follows that $\Prob{N \wedge M} \ge \left(1 - \frac{f^2(n)}{n-f(n)}\right)\left(1-\frac{1}{n}\right) \ge 1 - \frac{1}{n} - \frac{f^2(n)}{n-f(n)} \in 1 - o(1)$, as required.
\end{proof}

We next consider {\em potential cloud configurations}, namely,
$Z=\langle \sigma_0, \sigma_1, \ldots, \sigma_k \rangle$,
where $\sigma_i \in \Sigma$ for every $i$, and more generally,
{\em potential cloud configuration sequences} ${\bar Z}^r=(Z^1,\ldots,Z^r)$, 
where each $Z^i$ is a potential cloud configuration,
which may potentially occur as the configuration tuple of some 
influence clouds in round $i$ of some execution of Algorithm $A$
(in particular, the lengths of the cloud configurations $Z^i$ 
are monotonely non-decreasing).
We study the occurrence probability of potential cloud configuration sequences.

We say that the potential cloud configuration
$Z=\langle \sigma_0, \sigma_1, \ldots, \sigma_k \rangle$ {\em is realized}
by the initiator $u$ in round $r$ of execution $\alpha$ 
if the influence cloud ${\cal IC}^r_u = (C^r_u, S^r_u)$ 
has the same node states in $S^r_u$ as those of $Z$, or more formally, 
$S^r_u = \langle
\sigma_r(u,\alpha), \sigma_r(w_1,\alpha),\ldots, \sigma_r(w_k,\alpha) \rangle$,
such that $\sigma_r(u,\alpha) = \sigma_0$ and
$\sigma_r(w_i,\alpha) = \sigma_i$ for every $i\in [1,k]$.
In this case, the influence cloud ${\cal IC}^r_u$ is referred to as a
{\em realization} of the potential cloud configuration $Z$.
(Note that a potential cloud configuration may have many different 
realizations.)

More generally, we say that the potential cloud configuration sequence
${\bar Z}^r=(Z^1,\ldots,Z^r)$ is realized by the initiator $u$ 
in execution $\alpha$ if for every round $i = 1,\ldots, r$, the influence cloud
${\cal IC}^i_u$ is a realization of the potential cloud configuration $Z^i$.
In this case, the sequence of influence clouds of $u$ up to round $r$, 
${\bar{\cal IC}}^r_u = \langle {\cal IC}^1_u,\ldots,{\cal IC}^r_u\rangle$,
is referred to as a realization of ${\bar Z}^r$.
(Again, a potential cloud configuration sequence may have many different 
realizations.)

For a potential cloud configuration $Z$, let $E^r_u(Z)$ be the event that 
$Z$ is realized by the initiator $u$ in (round $r$ of) the run 
of algorithm $A$.
For a potential cloud configuration sequence ${\bar Z}^r$, let
$E_u({\bar Z}^r)$ denote the event that ${\bar Z}^r$ is realized 
by the initiator $u$ in (the first $r$ rounds of) the run of algorithm $A$.

\begin{lemma} 
\label{lem:clouds}
Restrict attention to executions of algorithm $A$ that satisfy event $N$, 
namely, in which all final influence clouds are disjoint.
Then $\Prob{E_u({\bar Z}^r)} = \Prob{E_v({\bar Z}^r)}$
for every $r \in [1,\rho]$, every potential cloud configuration sequence ${\bar 
  Z}^r$, and every two initiators $u$ and $v$.
\end{lemma}

\begin{proof}
The proof is by induction on $r$.
Initially, in round $1$, all possible influence clouds of algorithm $A$ are
singletons, i.e., their node sets contain just the initiator.
Neither $u$ nor $v$ have received any messages from other nodes.
This means that
$\Prob{\sigma_1(u)=s} = \Prob{\sigma_1(v)=s}$ for all $s\in \Sigma$,
thus any potential cloud configuration $Z^1=\langle s\rangle$ has 
the same probability of occuring for any initiator, implying the claim.

Assuming that the result holds for round $r-1\ge 1$, 
we show that it still holds for round $r$.
Consider a potential cloud configuration sequence
${\bar Z}^r=(Z^1,\ldots,Z^r)$ and two initiators $u$ and $v$.
We need to show that ${\bar Z}^r$ is equally likely to be realized 
by $u$ and $v$, conditioned on the event $N$.
By the inductive hypothesis, the prefix ${\bar Z}^{r-1}=(Z^1,\ldots,Z^{r-1})$
satisfies the claim. Hence it suffices to prove the following.
Let $p_u$ be the probability of the event $E^r_u(Z^r)$
conditioned on the event $N \wedge E_u({\bar Z}^{r-1})$.
Define
the probability $p_v$ similarly for $v$. 
Then it remains to prove that $p_u=p_v$.

To do that we need to show, for any state $\sigma_j \in Z^r$, that
the probability that $w_{u,j}$, the $j$th node in ${\cal IC}^r_u$, 
is in state $\sigma_j$, conditioned on the event $N \wedge E_u({\bar Z}^{r-1})$,
is the same as the probability that $w_{v,j}$, the $j$th node in ${\cal IC}^r_v$,
is in state $\sigma_j$, conditioned on the event $N \wedge E_v({\bar Z}^{r-1})$.

There are two cases to be considered. The first is that the potential 
influence cloud $Z^{r-1}$ has $j$ or more states. Then by our assumption
that events $E_u({\bar Z}^{r-1})$ and $E_v({\bar Z}^{r-1})$ hold, 
the nodes $w_{u,j}$ and $w_{v,j}$ 
were already in $u$'s and $v$'s influence clouds, respectively, 
at the end of round $r-1$. 
The node $w_{u,j}$ changes its state from its previous state, $\sigma'_j$,
to $\sigma_j$ on round $r$ as the result
of receiving some messages $M_1,\ldots,M_\ell$ from neighbors 
$x^u_1,\ldots,x^u_\ell$ in $u$'s influence cloud ${\cal IC}^{r-1}_u$, 
respectively. In turn, node $x^u_j$ sends message $M_j$ to $w_{u,j}$ on round $r$
as the result of being in a certain state $\sigma_r(x^u_j)$ at the beginning 
of round $r$ (or equivalently, on the end of round $r-1$) 
and making a certain random choice 
(with a certain probability $q_j$ for sending $M_j$ to $w_{u,j}$).
But if one assumes that the event $E_v({\bar Z}^{r-1})$ holds, namely, that
${\bar Z}^{r-1}$ is realized by the initiator $v$,
then the corresponding nodes
$x^v_1,\ldots,x^v_\ell$ in $v$'s influence cloud ${\cal IC}^{r-1}_v$
will be in the same respective states
($\sigma_r(x^v_j) = \sigma_r(x^u_j)$ for every $j$) on the end of round 
$r-1$, and therefore will send
the messages $M_1,\ldots,M_\ell$ to the node $w_{v,j}$ with the same 
probabilities $q_j$. Also, on the end of round $r-1$, the node $w_{v,j}$ is in 
the same state $\sigma'_j$ as $w_{u,j}$ (assuming event $E_v({\bar Z}^{r-1})$).
It follows that the node $w_{v,j}$ changes its state 
to $\sigma_j$ on round $r$ with the same probability as the node $w_{u,j}$.

The second case to be considered is when the potential influence cloud 
$Z^{r-1}$ has fewer than $j$ states. This means (conditioned on the events 
$E_u({\bar Z}^{r-1})$ and $E_v({\bar Z}^{r-1})$ respectively) that the nodes 
$w_{u,j}$ and $w_{v,j}$ were not in the respective influence clouds 
on the end of round $r-1$.
Rather, they were both passive nodes. 
By an argument similar to that made for round 1, any pair of (so far) 
passive nodes have equal probability of being in any state. Hence
$\Prob{\sigma_{r-1}(w_{u,j})=s} = \Prob{\sigma_{r-1}(w_{v,j})=s}$ 
for all $s\in \Sigma$.
As in the former case, the node $w_{u,j}$ changes its state from its previous 
state, $\sigma'_j$, to $\sigma_j$ on round $r$ as the result
of receiving some messages $M_1,\ldots,M_\ell$ from neighbors 
$x^u_1,\ldots,x^u_\ell$ that are already in $u$'s influence cloud 
${\cal IC}^{r-1}_u$, respectively. By a similar analysis, 
it follows that the node $w_{v,j}$ changes its state to $\sigma_j$ on round $r$ 
with the same probability as the node $w_{u,j}$.
\end{proof}

We now conclude that for every potential cloud configuration $Z$, 
every execution $\alpha$ and every two initiators $u$ and $v$,
the events $E^\rho_u(Z)$ and $E^\rho_v(Z)$ are equally likely. 
More specifically, we say that the potential cloud configuration $Z$ 
is \emph{equi-probable for initiators $u$ and $v$} if 
$\Prob{E^\rho_u(Z) \mid N} = \Prob{E^\rho_v(Z) \mid N}$.
Although a potential cloud configuration $Z$ may be the end-colud of many 
different potential cloud configuration sequences, and each such 
potential cloud configuration sequence may have many different realizations,
the above lemma implies the following (integrating over all possible choices).

\begin{corollary} 
\label{cor:clouds}
Restrict attention to executions of algorithm $A$ that satisfy event $N$, 
namely, in which all final influence clouds are disjoint.
Consider two initiators $u$ and $v$ and a potential cloud configuration $Z$.
Then $Z$ is equi-probable for $u$ and $v$.
\end{corollary}

By assumption, algorithm $A$ succeeds with probability at least $1/e + \eps$, for some fixed constant $\eps>0$.
Let $S$ be the event that $A$ elects exactly one leader.
We have
$$1/e + \eps \le \Prob{S} \le \Prob{S \mid M \wedge N} \Prob{M \wedge N} + \Prob{\text{not $(M\wedge N)$}}.$$
By Lemma~\ref{lem:noint}, we know that $\Prob{M \wedge N} \in 1 - o(1)$ and $\Prob{\text{not $(M\wedge N)$}} \in o(1)$, and thus it follows that
\begin{align}
  \Prob{S \mid M \wedge N} \ge \frac{{1}/{e} + \eps - o(1)}{1 - o(1)} > \frac{1}{e},
  \label{eq:succN}
\end{align}
for sufficiently large $n$.
By Cor.~\ref{cor:clouds}, each of the initiators has the same probability $p$ 
of realizing a potential cloud configuration where some node is a leader.
Assuming that events $M$ and $N$ occur, it is immediate that $0<p<1$.
Let $X$ be the random variable that represents the number of disjoint influence clouds.
Recall that algorithm $A$ succeeds whenever event $S$ occurs.
Its success probability assuming that $X=c$, at most $f(n)$ messages are sent, and all influence clouds are disjoint, is given by
\begin{align}
    \Prob{S \mid M \wedge N \wedge (X=c)} = c p (1-p)^{c-1}. \label{eq:succ}
\end{align}
For any given $c>0$, the value of \eqref{eq:succ} is maximized if
$p=\frac{1}{c}$, which yields that 
$\Prob{S \mid M \wedge N \wedge (X=c)} \le 1/e$ for any $c$.
It follows that $\Prob{S \mid M \wedge N} \le 1/e$ as well.
This, however, is a contradiction to \eqref{eq:succN} and completes the proof of
Theorem~\ref{thm:lb} for algorithms without unique identifiers.

We now argue why our result holds for any algorithm $B$ that assumes that 
nodes are equipped with unique ids (chosen by the adversary).
Let $S_{B}$ be the event that $B$ succeeds in leader election.
Suppose that $B$ sends only $f(n) \in o(\sqrt{n})$ messages with high probability but $\Prob{S_{B}} \ge 1/e + \eps$, for some constant $\eps > 0$.
Now consider an algorithm $B'$ that works in a model where nodes do not have ids.
Algorithm $B'$ is identical to $B$ with the only difference that before performing any other 
computation, every node generates a random number from the range $[1,n^4]$ and uses this
value in place of the unique id required by $B$.
Let $I$ be the event that all node ids are distinct; clearly  $\Prob{I}\ge 1 - 1/n$.
By definition of $B'$, we know that 
$\Prob{S_B} = \Prob{S_{B'} \mid I}$ and, from the anonymous case above, we get
$\Prob{S_{B'} \mid I}\Prob{I} \le \Prob{S_{B'}} \le 1/e + o(1),$ since only $o(\sqrt{n})$ messages are sent with high probability by $B'$.
It follows that
$\Prob{S_{B'} \mid I} \le \frac{1/e + o(1)}{1-1/n} < 1/e + o(1),$ and thus also $\Prob{S_B} \le 1/e + o(1)$, which is a contradiction.
This completes the proof of Theorem~\ref{thm:lb}.

\section{Conclusion}
\label{sec: conclusions}
We studied the role played by randomization in distributed leader election.
Some open questions on randomized leader election are raised by our 
work:
(1) Can we  improve the message complexity and/or running time for general 
graphs?
(2) Is there a separation between the message complexity for algorithms that succeed with high probability versus algorithms that achieve leader election with large constant probability?

\section*{References}

%
%
%
%
%

%

\end{document}